\documentclass[12pt,twoside,onecolumn,draftcls]{IEEEtran}

\usepackage{graphicx,color}

\usepackage{mathtools}

\usepackage{listings}
\usepackage{hyperref}
\usepackage{caption}
\usepackage{subcaption}
\usepackage{amsfonts,amsthm}

\usepackage{cite}

\newtheorem{theorem}{Theorem}

\begin{document}

\title{Spatial Field estimation from Samples taken at Unknown Locations generated by an Unknown Autoregressive Process}

\author{Sudeep Salgia and Animesh Kumar}

\maketitle

\begin{abstract}

Sampling physical fields with mobile sensors is an upcoming field of interest.
This offers greater advantages in terms of cost as often just a single sensor
can be used for field sensing, can be employed almost everywhere without sensing
stations, and has nominal operational costs. In a sampling scenario using mobile
sensors, knowing sampling locations more accurately leads to a manifold increase
in the costs. Moreover, the inertia of the moving vehicle constrains the
independence between the intersample distances, making them correlated. This
work, thus, aims at estimating spatially bandlimited fields from samples,
corrupted with measurement noise, collected on sampling locations obtained from
an autoregressive model on the intersample distances. The autoregressive model
is used to capture the correlation between the intersample distances. In this
setup of sampling at unknown sampling locations obtained from an autoregressive
model on intersample distances, the mean squared error between the field and its
estimated version has been shown to decrease as O($1/n$), where $n$ is the
average number of samples taken by the moving sensor.

\end{abstract}

\begin{IEEEkeywords}
Additive white noise, autoregressive model, nonuniform sampling, signal
reconstruction, signal sampling, wireless sensor networks.
\end{IEEEkeywords}

\section{Introduction}

Consider the scenario of sampling a spatially smooth field by using a sensor
mounted on a vehicle moving along a specific path \cite{Vett_samp_orig1,
Vett_samp_orig2}. This problem has been addressed in literature in various
scenarios and constraints. This problem reduces to the classical sampling and
interpolation problem (as described in \cite{Vett_samp_orig1, Vett_samp_orig2,
samp_interp_1, samp_interp_2, samp_interp_3} ), if the samples are collected on
{\em precisely known} locations in absence of any measurement noise. A more
generic version with precisely known locations, in presence of noise, both
measurement and quantization, has also been addressed (Refer
\cite{known_loc_1,known_loc_2, known_loc_3,known_loc_4, known_loc_5,
known_loc_6}). However, in practical scenarios, it is often difficult or
expensive to obtain the precise locations of the samples. Motivated on similar
lines, the problem of estimating fields from unknown locations has also been
studied \cite{unknown_loc}. It has been shown that the mean squared error in
estimating a spatially bandlimited field from measurement-noise affected field
samples that are collected on unknown spatial locations obtained from an unknown
renewal process decreases as $O\left(1/n\right)$ where $n$ is the average number
of samples. This work looks into a more realistic extension of the work in
\cite{unknown_loc}. 

The work \cite{unknown_loc} considers the intersample distances to be coming
from an unknown renewal process and hence are independently identically
distributed (i.i.d) random variables. This may not be an accurate model,
especially if the velocity of the vehicle considered to be smooth over the
mobile sampling period. This can be attributed to the large average sampling
densities used in such setups to decrease the error. A practical sampling setup
agnostic to location information often takes samples at reasonably regular time
intervals. Since velocity is likely to be smooth over a number of samples, we
expect that if an intersample distance was large, the next intersample distance
is also likely to be large, implying correlation between intersample distances.
A realistic model of intersample distances in mobile sampling should account for
the correlation between them. The primary motivation behind this work is to look
into the estimation of a field in a more realistic scenario to bolster the
practical viability of mobile sensing.

The field model is similar to the previous work wherein a spatially smooth,
temporally fixed, finite support field in a single dimension is assumed. The
single dimension has been assumed for mathematical tractability. A better two or
three dimensional model can be worked on out similar lines and has been left for
future work. The field has been assumed to be bandlimited to ensure spatial
smoothness. The sampling has been modelled using an autoregressive model on the
intersample distances, i.e. intersample distance at an instant is linearly
dependent on its previous intersample distances plus a stochastic
term\cite{percival1993spectral}. An autoregressive model of order $p$ would
imply that the intersample distance at an instant would be linearly dependent on
its `$p$' previous intersample distances. A model of order $1$ not only provides
a sufficiently generic model of mobile sensing but also ensures mathematical
tractability for a first exposition. Therefore, 
\begin{align}\label{define_AR_model}
X_i = \rho X_{i - 1} + Y_i
\end{align}
where $X_i$ is the $i^{\text{th}}$ intersample distance and $Y_i$ is the
stochastic term. $Y_i$'s are taken to be i.i.d variables coming from an unknown
renewal process, that is, the underlying distribution is not known. The
coefficient $\rho < 1$, is a positive constant assumed to be known. Typically,
it is suffices to know the estimate of this value.\ The $\rho$ models an element
of smoothness of velocity, while the renewal process is modelling the variability
due to various physical factors. The renewal process is assumed to have a finite
support and that governs the extent of variability. As $\rho$ decreases and
variance of the underlying distribution of the unknown renewal process
essentially boils down to the case considered in \cite{unknown_loc}. Also the
samples are assumed to be corrupted with additive independent noise. The only
statistics known about the noise are that the noise has a zero mean and a finite
variance. Again, oversampling, that is, a large average sampling density, will
be the key to reduce the mean square error between the original field and the
estimated field.
\begin{figure}[!htb]
\centering
\includegraphics[width =0.6\textwidth]{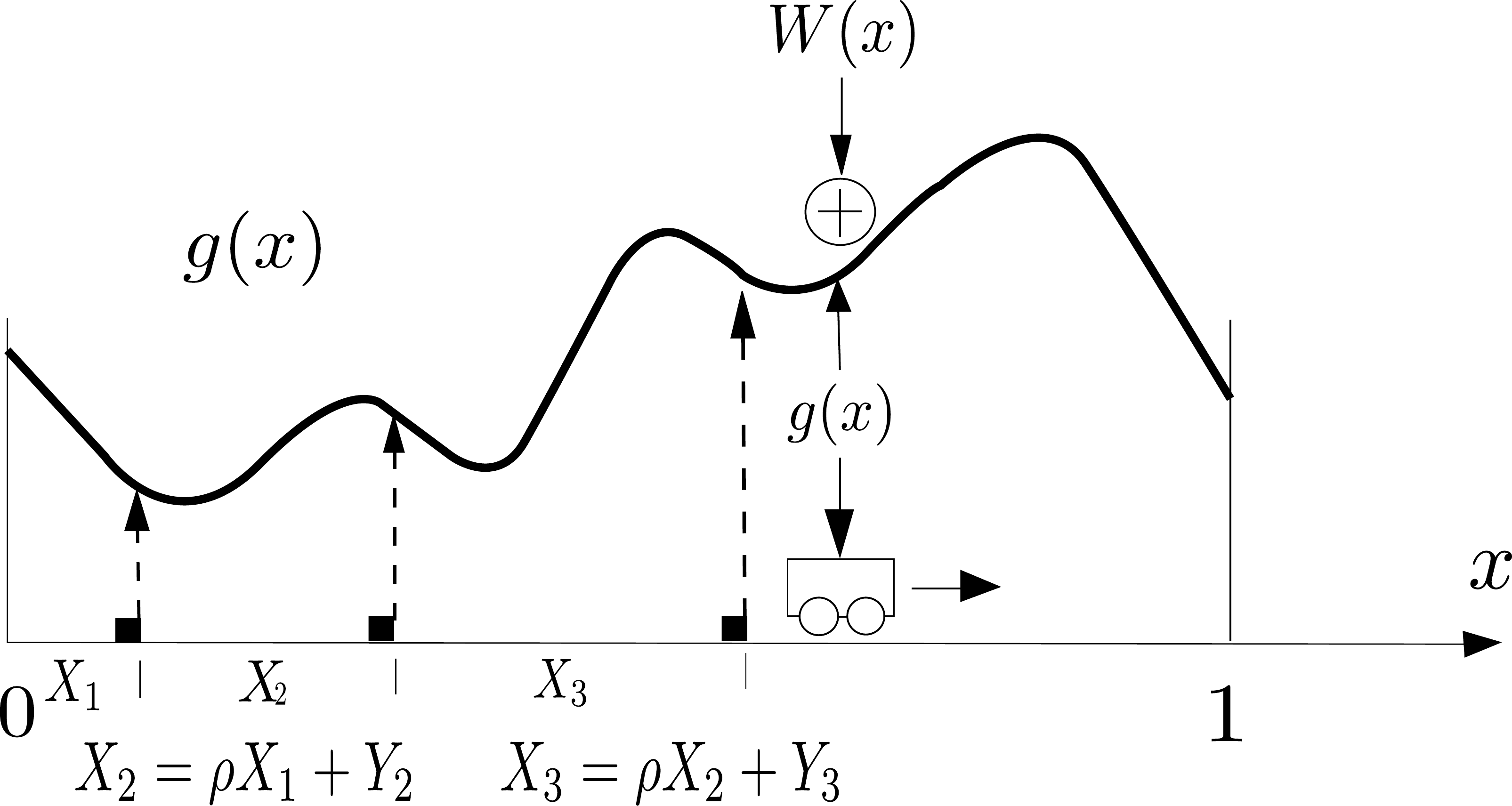}
\caption{\label{Fig:correlationfield} The mobile sampling scenario under study
is illustrated. A mobile sensor, where the field is temporally fixed, collects
the spatial field's values at unknown locations such that the intersample
distances are unknown and denoted by $X_1, X_2, \dots$. The figure shows that
how the intersample distances are modelled using an autoregressive model. Note
that how an initial large intersample distance leads to larger intersample
distances in future. It is also assumed that the samples are affected by
additive and independent noise process $W(x)$. Our task is to estimate $g(x)$
from the readings $g(s_1)+W(s_1), \ldots ,g(s_m)+W(s_m)$.}
\end{figure} 
The main result shown in the paper is that for sampling setup as described
above, the mean square error between the actual field and the estimated field
decreases as $ O\left(\frac{1}{n}.(1 - \rho^n) \right) \sim O(1/n)$ where,
$\rho$ is the correlation coefficient described above and $n$ is the average
sampling density. Note that the asymptotic bound $O(1/n)$, but the first term is
just used to show the effect of $\rho$ in the result. The result holds only for
beyond a certain sampling density threshold, the value of which is governed by
$\rho$ and is an increasing function of $\rho$. For reasonable values of $\rho$,
say $\rho = 0.9$, and $\lambda = 2$ this threshold is around $400$, which is a
small value in mobile sensing setups. Here, $\lambda$ is a parameter governing
the support of the unknown renewal process.

{\em Prior art}: Sampling and reconstruction of discrete-time bandlimited
signals from samples taken at unknown locations was first studied by Marziliano
and Vetterli \cite{prior_art_1} who had addressed the problem in a discrete-time
setup. Browning\cite{prior_art_2} later proposed an algorithm to recover
bandlimited signals from a finite number of ordered nonuniform samples at
unknown sampling locations. Nordio~et~al.~\cite{prior_art_3} studied the
estimation of periodic bandlimited signals, where sampling at unknown locations
is modelled by a random perturbation of equi-spaced deterministic grid. More
generally, the topic of sampling with jitter on the sampling locations
\cite{samp_interp_1}, \cite{prior_art_4},  is well known in the literature.
Mallick and Kumar \cite{prior_art_5} worked on reconstruction of bandlimited
fields from location-unaware sensors restricted on a discrete grid.  A more
generic case of sampling from unknown locations coming from a known underlying
distribution was introduced recently\cite{prior_art_6}. Further, the work
\cite{unknown_loc}, deals with estimation of field from unknown sampling
locations coming from an unknown renewal process. This work is different from
all others in the sense that the sampling model incorporates the correlation
between the intersample distances and thus addresses a more practical scenario.
The intersample distances are {\em unknown} have been considered to be coming
from an autoregressive model of order $1$ whose stochastic part is an {\em
unknown} renewal process.

{\em Notation}: All spatial fields under consideration which satisfy the set of
given assumptions are denoted by $g(x)$ and their corresponding spatial
derivatives will be denoted by $g'(x)$. $n$ denotes the average sampling
density, while $M$ is the random variable which denotes the number of samples
taken over the support of the field. The $\mathcal{L}^{\infty}$ norm of a
function will be denoted by $||g||_{\infty}$. The expectation operator will be
denoted by $\mathbb{E}$. The expectation is over all the random variables within
the arguments. The set of integers, reals and complex numbers will be denoted by
$\mathbb{Z}, \mathbb{R}$, and $\mathbb{C}$ respectively. Also, $j = \sqrt{-1}$.

{\em Organization}: The model of the spatial field, the distortion metric, the
sampling model explained in detail, elucidating the autoregressive model and the
measurement noise model will be discussed in Section II. The estimation of the
field from the samples has been elaborated in Section III. The Section~IV has the
simulations and Section V has conclusions and insights.

\section{Field Model, Distortion Metric, Sampling Model and Measurement Noise
Model}

The section describes the models which have been used for analysis in this work.
Firstly the field model is described followed by the distortion metric, the
sampling model and finally the measurement noise model.

\subsection{Field Model}

This work will assume that the field that is being sampled is one dimensional
spatially bandlimited signal and is temporally fixed. The bandlimitedness in
space ensures smoothness of the field of interest. The spatial dimension is $x
\in \mathbb{R}$ and the field is denoted by $g(x)$. The static nature in
temporal domain has been justified in \cite{Vett_samp_orig1, unknown_loc}. It is
a suitable assumption when the temporal variation in the field is far slower
than the speed of the mobile sensor. Furthermore, it helps make the analysis a
bit more tractable. The field is bounded and without loss of generality, $|g(x)|
\leq 1$. The function $g(x)$ has a Fourier series expansion given by
\begin{align}\label{def_basic_gx} g(x) = \sum_{k = -b}^b a[k] \exp(j 2\pi x)
\quad ; \quad a[k] = \int_{0}^{1} g(x) \exp(-j2 \pi x) dx \end{align} where $b$
is a known positive integer. Also, \begin{align}\label{bernestien} |g'(x)| \leq
2b\pi ||g||_{\infty} \leq 2b\pi \end{align} directly follows from the
Bernstein's inequality~\cite{bernstein}.

\subsection{Distortion Metric}

A simple and intuitive measure of the distortion will be used, the mean squared
error between the true field and its estimate. This will denote the energy of
the difference of the actual and the estimated signal. If $\hat{G}(x)$ is an
estimate of the field, then the distortion, $\mathcal{D}(\widehat{G}, g)$ is
defined as \begin{align} \mathcal{D}(\widehat{G}, g) := \mathbb{E} \left[
\int_0^1 |\widehat{G}(x) - g(x)|^2 dx \right] = \sum_{k = -b}^b \mathbb{E}
\left[|\hat{A}[k] - a[k]|^2 \right] \end{align} where, $\displaystyle \hat{A}[k]
= \int_{0}^{1} \widehat{G}(x) \exp(-j2 \pi x) dx $

\subsection{Sampling Model}

Let $ \displaystyle \{ X_i \}_{ i = 1}^M$ be the set of intersample distances
where $X_i$ is the distance between $i^{\text{th}}$ and $(i-1)^{\text{th}}$
sample and $X_1$ is the distance of the first sample from $x = 0$. The
intersample distances have been modelled as an autoregressive process of order
$1$ using a parameter $\rho$, $\displaystyle X_i = \rho X_{i-1} + Y_i \ \forall
\ i \geq 2$ and $X_1 = Y_1$. The $\rho$ models the dependence of the current
intersample distance on the previous one and $Y_i$ corresponds to the stochastic
term. $Y_i$'s are realised using an unknown renewal process. That is to say,
$Y_i$'s are independent and identically distributed random variables realized
from an {\em unknown} common distribution $Y$, such that $Y > 0$.  Using these
intersample distances, the sampling locations, $S_n$ are given by $\displaystyle
S_n = \sum_{i =1}^n X_i$. The sampling is done over an interval $[0,1]$ and $M$
is the random number of samples that lie in the interval i.e. it is defined such
that, $S_M \leq 1$ and $S_{M+1} > 1$. Thus $M$ is a well defined measurable
random variable~\cite{meas_var}.

For the purpose of ease of analysis and tractability, the support of the
distribution of $Y$ is considered to be finite and inversely proportional to the
sampling density. Hence, it is assumed that
\begin{align}
\label{distribution_bounds}
0 < Y \leq \frac{\lambda}{n} \text{  and  } \displaystyle \mathbb{E}[Y] =
\frac{1}{n},
\end{align}
where $\lambda > 1$ is a parameter that characterizes the support of the
distribution. It is a finite number and is independent of $n$. This would be a
crucial factor that governs the constant of proportionality in the expected mean
squared error in the estimate of the field. Furthermore, it is also an important
factor that determines the threshold on the minimum number of samples. Note
that, 
\begin{align}
X_1 = Y_1 ; \ \ X_2 = \rho X_1 + Y_2 = \rho Y_1 + Y_2 ; \ \ X_3 = \rho X_2 + Y_3
= \rho^2 Y_1 + \rho Y_2 + Y_3
\end{align}
which can be generalized as $\displaystyle X_n = \sum_{r = 1}^n \rho^{n-r} Y_r$.
This can be used to find a closed form expression of $S_n$ which can be written
as $\displaystyle S_n = \sum_{i = 1}^n X_i = \sum_{i = 1}^n \sum_{r = 1}^i
\rho^{i-r} Y_r$. Therefore, 
\begin{align}\label{S_n_simplified}
S_n = \sum_{i = 1}^n \sum_{r = 1}^i \rho^{i-r} Y_r = \sum_{i = 1}^n \sum_{r =
0}^{n - i} \rho^{r} Y_i = \dfrac{1}{1-\rho} \sum_{i = 1}^n (1 - \rho^{n-i+1})
Y_i = \dfrac{1}{1-\rho} \sum_{i = 1}^n c_{i,n} Y_i
\end{align}
where, $ c_{i,n} = 1 - \rho^{n-i+1}$, for $ i = {1, 2, 3 \dots n} $. For $0 \leq
\rho < 1$,
\begin{align}\label{c_bounds} 
\displaystyle \min_{i} c_{i,n} = 1 - \rho \text{ and } \displaystyle \max_{i}
c_{i,n} \leq 1 \ \text{for all } n.
\end{align}. 
We know that by definition, $S_{M+1} > 1$. Substituting $S_{M+1}$ from
\eqref{S_n_simplified},
\begin{align}
1 < \frac{1}{1- \rho} \sum_{i = 1}^{M+1} c_{i,M+1} Y_i  < \frac{(\max_{k}
c_{k,M+1})}{1- \rho}  \sum_{i = 1}^{M+1} Y_i \leq \frac{1}{1- \rho} \sum_{i =
1}^{M+1} Y_i \leq \frac{\lambda (M+1)}{n(1- \rho)} 
\end{align}
The last step follows from \eqref{distribution_bounds}. Therefore, 
\begin{equation}\label{M_upper}
M > \frac{n (1- \rho)}{\lambda} - 1
\end{equation}
This gives a lower bound on the value of $M$. To find the bounds on the expected
number of samples, $\mathbb{E}[M]$, the following lemma is noted.

{\em Lemma}: For the sampling model described as above with the intersample
distances coming from an autoregressive model, the average number of samples
taken in over the interval obeys the following bounds
\begin{equation*}
n(1 - \rho) - 1 \leq \mathbb{E}[M] \leq n + \frac{\lambda}{1 - \rho} - 1
\end{equation*}
{\em Proof}: For the proof of the lemma, we first need to consider the upper
bound on the value of $\mathbb{E}[S_{M+1}]$,
\begin{align}\label{upper_Sm1}
  \mathbb{E}[S_{M+1}] = \mathbb{E}\left[ \frac{1}{1- \rho} \sum_{i = 1}^{M+1}
c_{i,M+1} Y_i \right] =  \frac{1}{1- \rho}  \mathbb{E}\left[ \sum_{i = 1}^{M+1}
c_{i,M+1} Y_i\right]  \leq  \frac{(\max_{k} c_{k,M+1}) }{1- \rho}
\mathbb{E}\left[ \sum_{i = 1}^{M+1} Y_i\right]
\end{align}
Now using Wald's identity\cite{meas_var}, one can write,  
\begin{align}\label{wald}
\displaystyle \mathbb{E}\left[ \sum_{i = 1}^{M+1} Y_i\right] =
\mathbb{E}[Y]\mathbb{E}[M + 1] = \mathbb{E}[Y](\mathbb{E}[M] + 1)
\end{align} 
It is important to note that Wald's identity is applicable on this expression
and not directly on $S_{M+1}$ because $Y_i$'s are a set of independent and
identically distributed random variables while $X_i$'s are not. Using this with
\eqref{upper_Sm1} and the bounds obtained in \eqref{c_bounds}, we can write,
$\displaystyle \mathbb{E}[S_{M+1}] \leq  \frac{1}{1- \rho}
\left(\frac{1}{n}\right) (\mathbb{E}[M] + 1)$. Since, by definition, $S_{M+1} >
1$, therefore, $\mathbb{E}[S_{M+1}] > 1$. Combine this with the result in
\eqref{upper_Sm1} and \eqref{wald} to get, 
\begin{align}\label{EM_lower}
\mathbb{E}[M]  > n(1 - \rho) - 1
\end{align} 
Similarly we can consider a lower bound on $\mathbb{E}[S_{M+1}]$,
\begin{align}\label{ESM1_lower}
\mathbb{E}[S_{M+1}] \geq  \frac{1}{1- \rho} (\min_{k} c_{k,M+1})
\mathbb{E}\left[ \sum_{i = 1}^{M+1} Y_i\right] = \frac{1- \rho}{1- \rho}
\mathbb{E}[Y]. (\mathbb{E}[M] + 1) = \frac{1}{n} (\mathbb{E}[M] + 1)
\end{align}
Thus to upper bound the value of $\mathbb{E}[M]$, the following equation is
considered,
\begin{align*}
S_{M+1} - S_M  & = \dfrac{1}{1-\rho} \sum_{i = 1}^{M+1} (1 - \rho^{M-i+2}) Y_i -
\dfrac{1}{1-\rho} \sum_{i = 1}^M (1 - \rho^{M-i+1}) Y_i \\
& = \dfrac{1}{1-\rho} \sum_{i = 1}^{M+1} (\rho^{M-i+1} - \rho^{M-i}) Y_i =
\sum_{i = 1}^{M+1} \rho^{M-i+1} Y_i \\
\end{align*}
Since, $S_M \leq 1$ and $Y \leq \dfrac{\lambda}{n}$, we can write, 
\begin{align}
S_{M+1} = S_M + \sum_{i = 1}^{M+1} \rho^{M-i+1} Y_i  \leq 1 + \sum_{i = 1}^{M+1}
\rho^{M-i+1} \frac{\lambda}{n} \leq 1 + \frac{\lambda}{n} \sum_{i = 0}^{\infty}
\rho^{i} = 1 + \frac{\lambda}{n(1-\rho)} 
\end{align}
This implies, 
\begin{equation}\label{ESM1_upper}
\mathbb{E}[S_{M+1}] \leq 1 + \frac{\lambda}{n(1-\rho)}
\end{equation}
Combine equation \eqref{ESM1_upper} with equations \eqref{ESM1_lower} and
\eqref{EM_lower} to obtain, 
\begin{align}\label{EM_bounds}
 n(1 - \rho) - 1 \leq \mathbb{E}[M] \leq n + \frac{\lambda}{1 - \rho} - 1
\end{align}
This completes the proof of the lemma.  Since the expected number of samples is
of the order of $n$, therefore, $n$ is termed the sampling density. However, the
results are governed by the {\em effective} sampling density which is given by
$n(1 - \rho)$. The difference becomes relevant at values of $\rho$ close to $1$
and finitely large $n$. More detailed explanation about this has been given in
Section III, in the light of obtained results.

\subsection{Measurement Noise model}
It will be assumed that the obtained samples have been corrupted by additive
noise that is independent both of the samples and of the renewal process. Thus
the samples obtained would be sampled versions of $g(x) + W(x)$, where $W(x)$ is
the noise. Also, since the measurement noise is independent, that is for any set
of measurements at distinct points $s_1, s_2, s_3, \dots s_n$, the samples
$W(s_1), W(s_2), W(s_3), \dots, W(s_n)$ would be independent and identically
distributed random variables. Thus the sampled version of the measurement noise
has been assumed to be a discrete-time white noise process. It is essential to
note that the distribution of the noise is also {\em unknown}. The only
statistics known about that noise in addition to the above are that the noise is
zero mean and has a finite variance, $\sigma^2$.

\section{Field Estimation from the Obtained Samples}

The primary idea in the reconstruction of the field would be the estimation of
the Fourier coefficients of the field using the noisy samples that have been
obtained at unknown sampling locations where the intersample distances have been
modelled using an autoregressive model. The approach is similar
to\cite{unknown_loc}, however, there would be a difference in the analysis
arising because of the correlated intersample distances. A Riemann sum kind of
approximation using the obtained samples will be used to get estimates of the
Fourier coefficients of the signals. Define the estimate of the Fourier
coefficients as 
\begin{align}\label{AK_gen}
  \hat{A}_{\text{gen}} [k] = \frac{1}{M} \sum_{i = 1}^M \big[ g(S_i) + W(S_i)]
\exp\left(-\frac{j 2 \pi ki}{M} \right)
\end{align}
The motivation to use this as the estimate is the same as the one in the
previous paper. The difference is in how this estimate behaves. This is because
in \cite{unknown_loc}, the sample locations were considered to be ``near'' the
grid points as the sampling was on locations obtained from i.i.d.~intersample
distances. Thus, because of independence, each point was individually likely to
be close to the grid points, unaffected by others. However, in this case, the
intersample distances are correlated and hence the error in one will propagate
to all the further ones breaking the premise of these locations being ``near''
to the grid points. Despite this error propagation, it will be shown that the
bound on error is still $O(1/n)$ which is non-trivial. Thus, even though the
work is inspired from the approach taken in \cite{unknown_loc}, the analysis in
the given setup is far more challenging.  With our estimate under this scenario,
the bound on the mean square error is analysed and the following theorem is
noted.
\begin{theorem}
 Define the estimate of the Fourier coefficients $\hat{A}_{\text{gen}} [k]$ as
in \eqref{AK_gen}. Then in the scenario that the intersample distances are
obtained from an autoregressive model, the expected mean squared error between
the estimated and the actual Fourier coefficients is upper bounded as 
\begin{align}
\mathbb{E}\left[ \big|\hat{A}_{\text{gen}}[k] - a[k] \big|^2 \right] \leq
\frac{C - C' \rho^n}{n}, 
 \end{align}
 where $C, C'$ are finite positive constants independent of $n$. The constants
depend on $\lambda$, a finite constant independent of $n$ that characterises the
support of the distribution of $Y$, the bandwidth $b$ of the field, the
`correlation coefficient' $\rho$ and the variance of the measurement noise,
$\sigma^2$.
\end{theorem}

\begin{proof}
The proof will be coupled step approach by using a Riemann sum approximation of
the integral in \eqref{def_basic_gx} and bounding the error between the Riemann
sum and the estimate of the Fourier coefficients given in \eqref{AK_gen}. The
detailed calculations will be separately shown in the Appendices while using the
main results to maintain the flow of the proof. Define,
\begin{equation*}
\hat{A}[k] := \frac{1}{M} \sum_{i = 1}^M g(S_i)\exp \left(-\frac{j 2 \pi ki}{M}
\right) \ ; \ \hat{W}[k] := \frac{1}{M} \sum_{i = 1}^M W(S_i)\exp \left(-\frac{j
2 \pi ki}{M} \right)
\end{equation*}
Therefore, $\hat{A}_{\text{gen}}[k] = \hat{A}[k] + \hat{W}[k]$. The above terms
are signal and noise components in the estimates and they have been separated to
ease out calculations. The integral in \eqref{def_basic_gx}, defining $a[k]$,
can be approximated using an $M$-point Riemann sum as, 
\begin{align}
A_R[k] = \frac{1}{M} \sum_{i = 1}^M g\left(\frac{i}{M} \right)\exp
\left(-\frac{j 2 \pi ki}{M} \right)
\end{align}
It is important to note here that is this sum is actually random because of $M$
and thus when used to calculate the distortion, the expected value of the sum
should be taken to average over different estimates of the fields.
\begin{align}\label{main_eq}
  \mathcal{D}(\widehat{G},g) & = \mathbb{E}\left[ \big| \hat{A}_{\text{gen}}[k]
- a[k] \big|^2 \right] \nonumber \\
   & = \mathbb{E}\left[ \big| \hat{A}[k] + \hat{W}[k] - a[k] \big|^2 \right]
\nonumber \\
   & \leq 2\mathbb{E}\left[ \big| \hat{A}[k] - a[k] \big|^2 \right] +
2\mathbb{E}\left[ \big| \hat{W}[k] \big|^2 \right] \nonumber \\
   & = 2\mathbb{E}\left[ \big| \hat{A}[k] -A_R[k] + A_R[k] - a[k] \big|^2
\right] + 2\mathbb{E}\left[ \big| \hat{W}[k] \big|^2 \right] \nonumber \\
   & \leq 4\mathbb{E}\left[ \big| \hat{A}[k] -A_R[k] \big|^2 \right] +
4\mathbb{E}\left[ \big| A_R[k] - a[k] \big|^2 \right] + 2\mathbb{E}\left[ \big|
\hat{W}[k] \big|^2 \right] 
\end{align}
where the second and the fourth step follow from the inequality $ \displaystyle
\big| \sum_{i =1}^n a_i \big|^2 \leq n\sum_{i=1}^n |a_i|^2$, for any numbers
$a_1, a_2, \dots, a_n \in \mathbb{C}$. The three terms as obtained in the final
step of \eqref{main_eq} obtained will be analysed separately and solved one by
one. The three terms are considered in their specific order.
\begin{align}
\big| \hat{A}[k] - A_R[k] \big| &= \bigg| \frac{1}{M} \sum_{i = 1}^M g(S_i)\exp
\left(-\frac{j 2 \pi ki}{M} \right) - \frac{1}{M} \sum_{i = 1}^M
g\left(\frac{i}{M} \right)\exp \left(-\frac{j 2 \pi ki}{M} \right) \bigg|
\nonumber \\
& = \frac{1}{M} \bigg|  \sum_{i = 1}^M \left[ g(S_i) -  g\left(\frac{i}{M}
\right) \right] \exp \left(-\frac{j 2 \pi ki}{M} \right) \bigg| \nonumber \\
& \leq \frac{1}{M}   \sum_{i = 1}^M \bigg| \left[ g(S_i) - g\left(\frac{i}{M}
\right) \right] \exp \left(-\frac{j 2 \pi ki}{M} \right) \bigg| \nonumber \\ 
& = \frac{1}{M}   \sum_{i = 1}^M \bigg|  g(S_i) -  g\left(\frac{i}{M} \right)
\bigg| 
\end{align}
The third step follows from the triangle inequality. Squaring the expression
obtained in the above equation, 
\begin{align}\label{si_im}
\big| \hat{A}[k] - A_R[k] \big|^2 &=  \frac{1}{M^2} \bigg\{ \sum_{i = 1}^M
\bigg|  g(S_i) -  g\left(\frac{i}{M} \right) \bigg| \bigg\}^2 \nonumber \\
& \leq \frac{M}{M^2}  \sum_{i = 1}^M \bigg|  g(S_i) -  g\left(\frac{i}{M}
\right) \bigg|^2  \nonumber \\
& = \frac{1}{M}  \sum_{i = 1}^M \bigg|  g(S_i) -  g\left(\frac{i}{M} \right)
\bigg|^2  \nonumber \\
& \leq \frac{1}{M}  \sum_{i = 1}^M ||g'||_{\infty}^2 \bigg|  S_i - \frac{i}{M}
\bigg|^2
\end{align}
The second step follows from the inequality stated before, i.e., $  \left(
\sum_{i =1}^n a_i \right)^2 \leq n\sum_{i=1}^n a_i^2$, for any numbers $a_1, a_2
\dots a_n \in \mathbb{R}$. The last step uses the smoothness property of the
field $g(x)$. For any smooth field $g(x)$ over a domain $\mathcal{D} \subset
\mathbb{R}$ and any $x_1, x_2 \in \mathcal{D}$, $|g(x_1) - g(x_2)| \leq
||g'|||_{\infty} |x_1 - x_2|$. This follows from the Lagrange's mean value
theorem. To get the first term in \eqref{main_eq}, we need to take the
expectations on either side in \eqref{si_im}. Therefore,
\begin{align}
  \mathbb{E}\left[ \big| \hat{A}[k] - A_R[k] \big|^2 \right] = ||g'||_{\infty}^2
\mathbb{E} \left[ \frac{1}{M}  \sum_{i = 1}^M  \bigg|  S_i - \frac{i}{M}
\bigg|^2\right]
\end{align}
The expectation in the right hand side has been calculated in detail in Appendix
A and from the result obtained there, 
\begin{align}\label{bound_1}
\mathbb{E} \left[ \frac{1}{M}  \sum_{i = 1}^M  \bigg|  S_i -  \frac{i}{M}
\bigg|^2\right] \leq \frac{C_0 (1 - C_1 \rho^n)}{n}
\end{align}
which is independent of the distribution of the renewal process and only depends
on the support parameter $\lambda$ and $C_0, C_1$ are constants independent of
$n$. The important part here is to note that the bound is guaranteed as the
average sampling density, $n$, becomes large and more specifically if it is
atleast $\displaystyle \frac{\lambda}{1- \rho} \left(1 - \frac{2}{\ln \rho}
\right)$. Note that since $0 < \rho < 1$, therefore, $\ln \rho < 0$. Thus, the
bound is always a positive number. In fact, since $\ln x$ is an increasing
function of x, so is this threshold value an increasing function of $\rho$. It
is important to note that this is a sufficient condition and not a necessary
condition. Also, this value is relatively small for mobile sensing setups. For
examples, for $\rho = 0.9$ and $\lambda = 2$, this value is roughly about $400$.

The other two terms obtained in equation \eqref{main_eq} are exactly the same as
the ones obtained in \cite{unknown_loc} and reproducing the bounds obtained
there (Appendix B and equation (15)), we can write,
%
\begin{align}\label{old}
\mathbb{E}\left[ \big| A_R[k] - a[k] \big|^2 \right] \leq \mathbb{E}\left[
\frac{16b^2\pi^2}{M^2} \right] \\
\mathbb{E}\left[ \big| \hat{W}[k] \big|^2 \right] \leq
\mathbb{E}\left[\frac{\sigma^2}{M}\right] 
%
%
\end{align}
Combining these results with the one obtained in equation \eqref{M_upper}, we
get
\begin{align}\label{old_final}
\mathbb{E}\left[ \big| A_R[k] - a[k] \big|^2 \right]  \leq
\frac{16b^2\pi^2\lambda^2}{(n(1 - \rho) - \lambda)^2} \\
\mathbb{E}\left[ \big| \hat{W}[k] \big|^2 \right] \leq \frac{\sigma^2
\lambda}{n(1 - \rho) - \lambda} 
\end{align}
Putting together the results obtained in \eqref{bernestien}, \eqref{main_eq},
\eqref{bound_1}, \eqref{old_final}, 
\begin{align}
\mathbb{E}\left[ \big| \hat{A}_{\text{gen}}[k] - a[k] \big|^2 \right] & \leq 4
||g'||_{\infty}^2 \frac{C_0 (1 - C_1 \rho^n)}{n} +
\frac{64b^2\pi^2\lambda^2}{(n(1 - \rho) - \lambda)^2} + \frac{2\sigma^2
\lambda}{n(1 - \rho) - \lambda} \nonumber \\ 
& \leq  4 (2b\pi)^2 \frac{C_0 (1 - C_1 \rho^n)}{n} +
\frac{64b^2\pi^2\lambda^2}{(n(1 - \rho) - \lambda)^2} + \frac{2\sigma^2
\lambda}{n(1 - \rho) - \lambda} \nonumber\\
& \leq \frac{C - C' \rho^n}{n}
\end{align}
as $n$ becomes large. $C, C'$ are positive constants independent of the sampling
density $n$. These are mainly functions of the bandwidth parameter $b$, the
support parameter $\lambda$, the `correlation coefficient' $\rho$ and the
measurement noise $\sigma^2$.

This completes the proof of the main result in the theorem.

\end{proof}

The dependence on the coefficient $\rho$ that characterizes the autoregressive
model is rather interesting and deserves special attention. On a broad scale,
due to the autoregressive model, two things have significantly changed from the
result shown in \cite{unknown_loc}. Firstly, it has resulted in a lower bound on
the average sampling density $n$, that is sufficient to ensure the bound on the
mean squared error. Even though it is not a necessary condition, the bound fails
when $n \ll \dfrac{\lambda}{1- \rho} \left(1 - \dfrac{2}{\ln \rho} \right)$.
That is, if sampling density is made too small, the aberration is visible and
the error does not go down as $1/n$. This can be seen in the simulations section
of the paper. Moreover, this threshold gets larger as $\rho$ gets closer to $1$
and in fact is unbounded and shoots off to infinity when $\rho$ is in the
neighbourhood of $1$. Another factor is the presence of the term of $1 - \rho^n$
in the numerator. This does not affect the asymptotic bounds on the error
because as $n$ becomes larger, $\rho^n$ becomes smaller since $0 \leq \rho 1$. 

Another thing that is worth noting that even though the stochastic process is
such that $n$ is expected sampling density, the {\em effective} sampling density
is lesser due to the autoregressive model. Even though $n$ has been termed as
sampling density for the ease of understanding, it is just that the sampling
density is of the order of $n$. The {\em actual} or {\em effective} sampling
density is $n(1 - \rho)$. Asymptotically, it cannot be made this is same as $n$,
however, for finitely large $n$, a clear difference is noted when $\rho$ is in
the neighbourhood of $1$. Consider the case of $n = 10000$, which is a
reasonably large number of samples and one expects almost perfect reconstruction
at this sampling density. If $\rho$ is close to $1$, say, $0.99$, the value of
{\em effective} sampling density becomes about $100$ which is not sufficient for
a good reconstruction. Thus, because of this {\em effective} sampling densities,
the error begins to decrease as $1/n$ from much larger values of $n$ when $\rho$
is close to $1$ as opposed to when it is not.

\section{Simulations}

This section presents the results of simulations. The simulations have been
presented in Fig.~\ref{Fig :Simulations}. The simulations have a large scope as
they help confirm the results obtained, help in examining the effect of the
`correlation coefficient', $\rho$ on the results and help in analysing the
effect of different renewal processes that characterize the stochastic term of
the autoregressive model on the estimation error.

Firstly, for the purpose of simulations, a field $g(x)$ with $b = 3$ is
considered and its Fourier coefficients have been generated using independent
trials of a Uniform distribution over $[-1,1]$ for all real and imaginary parts
separately. To ensure that the field is real, conjugate symmetry is employed,
i.e., $a[k] = \bar{a}[-k]$ and $a[0]$ is ensured that it is real. Finally, the
field is scaled to have $|g(x)| \leq 1$. The following coefficients were
obtained.
\begin{align}
a[0] = 0.3002; \ a[1] = -0.04131 + j0.0216; \nonumber \\
a[2] = 0.0871 + j0.0343; \ a[3] = -0.1679 - j0.0586; 
\end{align}
The distortion was estimated using Monte Carlo simulation with $10000$ trials.

\begin{figure}[!htb]
\centering
\includegraphics[width =6.0in]{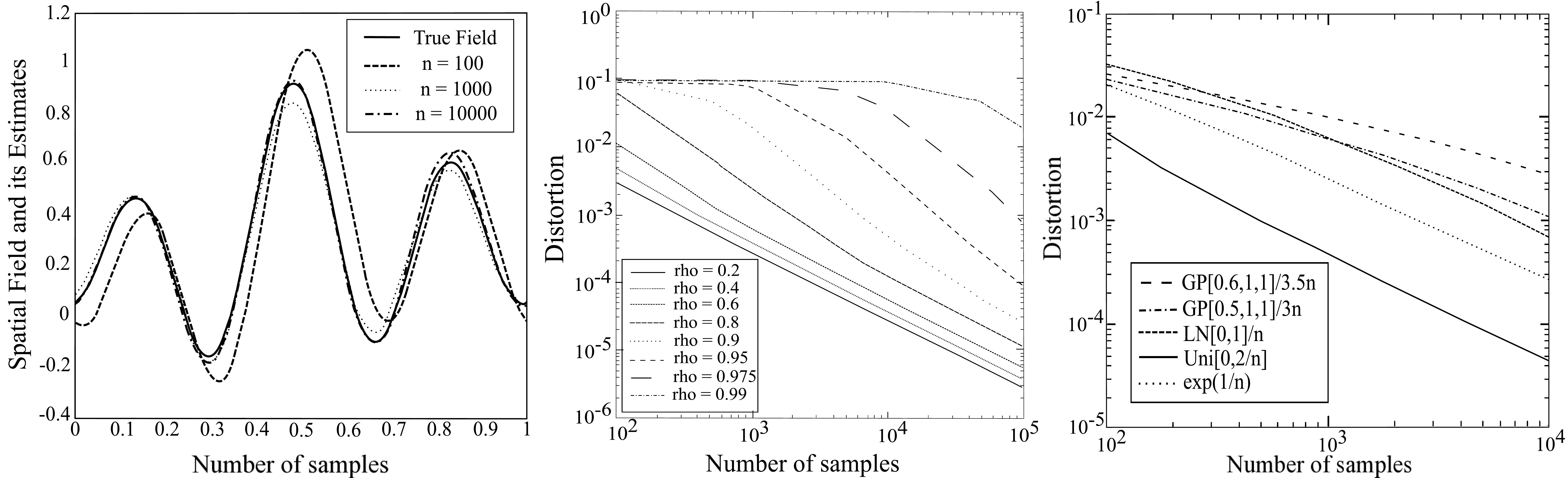}
\caption{\label{Fig :Simulations} (a) The convergence of random realizations of
$\widehat{G}(x)$ to $g(x)$ is illustrated for the field whose Fourier
coefficients have been mentioned above. The field estimate for $n = 10000$
almost converges to the true value of the field (and is not visible in the
plot). (b) The distortion scales as $O(1/n)$ with additive Gaussian Noise of
variance $0.125$ and Uniform$[0, 2/n]$ renewal process distribution. Note the
variation with $\rho$. For small values of $\rho$, a slope of $-1$ is distinctly
noted. However, as the value of $\rho$ increase, we see that the slope does not
initially follow the trend. This can be accounted for in two ways. One being
that at large values of $\rho$, the threshold is not being satisfied at these
sampling densities and the other is the {\em effective} sampling rate. Note that
how the slopes are eventually getting close to $-1$ as the sampling density is
increased. It be shown that at very large sampling densities the plot will
behave as expected.}
\end{figure} 

Field estimates were obtained in various cases using different distributions for
the renewal process and the measurement noise distributions. Uniform
distribution between $a$ and $b$ is denoted as Uniform$[a,b]$. The $\exp(\mu)$
denotes an exponential distribution with mean $1/\mu$. $\mathcal{N}(\mu,s)$
denotes a Gaussian distribution with mean $\mu$ and variance $s$. The
distribution LN$[m,s]$ represents a log-normal distribution where $m$ is the
mean and $s$ is the variance of the underlying Gaussian.

In Fig.\ref{Fig :Simulations}(a) the measurement noise was generated using
$\mathcal{N}(0,0.125)$. The autoregressive model that was employed had $\rho =
0.5$ and the stochastic part was generated using Uni$[0,2/n]$. Therefore,
$\lambda = 2$. The method for reconstruction in \eqref{AK_gen} is agnostic to
this distribution though. The figure has random realizations of the estimated
field $\hat{G}(x)$ along with the true field, $g(x)$. It can be clearly seen
that as the value of $n$ increases the estimated field gets closer to the
actual field.

The Fig.\ref{Fig :Simulations}(b) has the plots of the mean squared error
between the field and its estimate. For the plots, the noise model and
stochastic part of the autoregressive model is same as in above paragraph.
However, the plots have been plotted for different values of $\rho$ in the range
$0.2$ to $0.99$. Note that the illustrated plots are log-log plots.

The Fig.\ref{Fig :Simulations}(c) has plots for different renewal processes for
the stochastic part of the autoregressive process. For heavy-tailed, infinite
support distributions like the General Pareto distribution, the error does not
decrease as $1/n$. For the other finite support distributions, the error follows
the $O(1/n)$ decrease.

\section{Conclusion} 
\label{sec:conclusion}

The problem of sampling with unknown sampling locations obtained from an
autoregressive model on the intersample distances was studied. The field was
estimated using the noisy samples and the mean squared error between the
estimated and true field was shown to decrease as $O(1/n)$, where $n$ was the
order of the average sampling density. Also, the effect of correlated
intersample distances with different values of correlation was studied.
Simulation results, consonant with the theoretical analysis, were also
presented.

\section*{Appendix A} 

This section will mainly be elaborating the bound on $\displaystyle \mathbb{E}
\left[ \frac{1}{M}  \sum_{i = 1}^M  \bigg|  S_i -  \frac{i}{M} \bigg|^2\right]$.
To begin with consider the remainder term $R_M$, which is the distance remaining
after $M$ samples. Therefore, $R_M  = 1 - S_M$. An upper bound on this will be
established as this as this will be of use later. 
\begin{align*}
R_M = 1 - S_M < S_{M+1} - S_M = \sum_{i = 1}^{M+1} \rho^{M-i+1} Y_i \leq \sum_{i
= 0}^{\infty} \rho^{i} \frac{\lambda}{n} = \frac{\lambda}{n(1- \rho)}
\end{align*}
Thus, we have,
\begin{equation}
R_M \leq \frac{\lambda}{n(1- \rho)}
\end{equation}
The main expression that this Appendix will be dealing with is the upper bound
on the mean square error between the sampling locations and the uniform grid
locations. The expression is given by 
\begin{equation*}
\mathbb{E}\left[ \frac{1}{M} \sum_{i = 1}^M \bigg| S_i - \frac{i}{M}
\bigg|^2\right] = \mathbb{E}\left[ \frac{1}{M} \sum_{i = 1}^M \bigg| \sum_{k =
1}^i X_k - \frac{i}{M} \bigg|^2\right] = \mathbb{E}\left[ \frac{1}{M} \sum_{i =
1}^M \bigg\{\sum_{k = 1}^i  \left( X_k - \frac{1}{M} \right) \bigg\}^2\right]
\end{equation*}
To ease out the calculations, a conditional version of this expectation will be
dealt with that is, for $M = m$ which is given by $\displaystyle
\mathbb{E}\left[ \frac{1}{M} \sum_{i = 1}^M \bigg\{\sum_{k = 1}^i  \left( X_k -
\frac{1}{M} \right) \bigg\}^2 \bigg| M = m\right]$ and then later the assumption
will be relaxed for the final results. Define the following,
\begin{align}\label{bi_ci}
b(i) := \mathbb{E}\left[ \left( X_i - \frac{1}{m} \right) \bigg| M = m \right] =
\sum_{r = 1}^i \rho^{i-r} \mathbb{E}[Y_r | M = m] - \frac{1}{m} \\
c(i) := \text{Var}\left[ \left( X_i - \frac{1}{m} \right) \bigg| M = m \right] =
\text{Var}[X_i | M = m] 
\end{align}
Since we know that for a random variable $X$, $\mathbb{E}[X^2] = \text{Var}[X] +
\mathbb{E}[X]^2$. Therefore, we can write, 
\begin{equation}
\mathbb{E}\left[ \left( X_i - \frac{1}{m} \right)^2 \bigg| M = m \right]  = c(i)
+ b^2(i)
\end{equation}
Consider,
\begin{align}
\mathbb{E}[X_i X_j | M = m] & =   \mathbb{E} \left[\left( \sum_{k = 1}^i
\rho^{i-k} Y_k \right) \left( \sum_{l = 1}^j \rho^{j-l} Y_l \right) \bigg| M =
m\right]   \nonumber  \\
& = \mathbb{E} \left[\left( \sum_{k = 1}^i \rho^{i-k} \sum_{l = 1}^j \rho^{j-l}
Y_k Y_l \right) \bigg| M = m\right] \nonumber \\
& = \sum_{k = 1}^i  \sum_{l = 1}^j \rho^{i-k}\rho^{j-l} \mathbb{E}[Y_k Y_l | M =
m] \nonumber \\
& = \sum_{l = 1}^{\min\{i,j\}} \rho^{i+j-2l} \mathbb{E}[Y_l^2 | M = m]  +
2\sum_{k = 1}^i  \sum_{l = k+1}^j \rho^{i-k+j- l}\ \mathbb{E}[Y_k | M = m]
\mathbb{E}[Y_l | M = m] \nonumber \\
& =  \sum_{l = 1}^{\min\{i,j\}} \rho^{i+j-2l} \text{Var}[Y_l | M = m]  + \sum_{k
= 1}^i  \sum_{l = 1}^j \rho^{i-k + j-l} \mathbb{E}[Y_k|M = m] \mathbb{E}[Y_l | M
= m] \nonumber \\
& = \rho^{|i-j|} \sum_{l = 1}^a \rho^{2(a-l)} \text{Var}[Y_1 | M = m]+
\left(b(i) + \frac{1}{m} \right) \left(b(j) + \frac{1}{m} \right) \nonumber  \\
& = \rho^{|i-j|} c(a) + \left(b(i) + \frac{1}{m} \right) \left(b(j) +
\frac{1}{m} \right)   
\end{align}
where $a = \min \{i,j\}$ and $b(i)$ and $b(j)$ have been defined in
(\ref{bi_ci}).

\begin{flalign}
\label{xi_xj}
& \mathbb{E}\left[ \left( X_i - \frac{1}{m} \right)\left( X_j - \frac{1}{m}
\right)\bigg| M = m \right]  \nonumber \\
& = \mathbb{E}[ X_i X_j | M = m] - \frac{1}{m}(\mathbb{E}[X_i | M = m] +
\mathbb{E}[X_j | M = m])  + \frac{1}{m^2} \nonumber \\
& = \rho^{|i-j|} c(a) + \left(b(i) + \frac{1}{m} \right) \left(b(j) +
\frac{1}{m} \right) - \frac{1}{m}\left(b(i) + \frac{1}{m} + b(j) + \frac{1}{m}
\right) + \frac{1}{m^2} \nonumber \\
& = \rho^{|i-j|} c(\min\{i,j\}) + b(i)b(j)  
\end{flalign}

Define for convenience of notation, 
\begin{align}\label{tm_def}
T_m = \mathbb{E}\left[ \frac{1}{M} \sum_{i = 1}^M \bigg| S_i - \frac{i}{M}
\bigg|^2\right] \ ; \ Z_k = \left( X_k - \frac{1}{m} \right)
\end{align}
Therefore we can write, 
\begin{align*}
T_m &= \mathbb{E}\left[ \frac{1}{M} \sum_{i = 1}^M \bigg\{\sum_{k = 1}^i  \left(
X_k - \frac{1}{M} \right) \bigg\}^2 \bigg| M = m\right] = \mathbb{E}\left[
\frac{1}{m} \sum_{i = 1}^m \bigg\{\sum_{k = 1}^i  \left( X_k - \frac{1}{m}
\right) \bigg\}^2 \bigg| M = m\right]  \\
&= \mathbb{E}\left[ \frac{1}{m} \sum_{i = 1}^m \bigg\{\sum_{k = 1}^i  Z_k
\bigg\}^2 \bigg| M = m\right] = \mathbb{E}\left[ \frac{1}{m} \sum_{i = 1}^m (Z_1
+ Z_2 + \dots Z_i)^2 \bigg| M = m\right]  \\
& = \mathbb{E} \left[ \frac{1}{m} \left( mZ_1^2 + (m-1)Z_2^2 + \dots + Z_m^2 +
2\{ (m-1)Z_1Z_2 + (m-2)Z_1Z_3 + \dots + Z_1Z_m    \right. \right. \\
& \ \  \left.  + (m-2)Z_2Z_3 + (m-3) Z_2Z_4 + \dots Z_2Z_4 +\cdots + Z_{m-1}Z_m
\} \right) \bigg| M = m \bigg] \\
&  = \mathbb{E} \left[ \frac{1}{m} \left( 2\{ mZ_1^2 + (m-1)Z_2^2 + \dots +
Z_m^2 +  (m-1)Z_1Z_2 + (m-2)Z_1Z_3 + \dots + Z_1Z_m +  \right. \right. \\
& \left. \left. \ \ \ (m-2)Z_2Z_3 + (m-3)Z_2Z_4 + \cdots + Z_{m-1}Z_m \} -
\{mZ_1^2 + (m-1)Z_2^2 + \dots + Z_m^2 \} \right) \bigg| M = m \right]
\end{align*}
%
%
  
%
%
The similar terms can be grouped together to form a simpler expression like,
\begin{align*}
T_m & = \mathbb{E} \left[ \frac{1}{m} \left( 2\{ mZ_1^2 + (m-1)Z_1Z_2 +
(m-2)Z_1Z_3 + \dots + Z_1Z_m + \right. \right.\\ 
& \ \ \ \ \ \ \ \ \ \ (m-1)Z_2^2 + (m-2)Z_2Z_3 + (m-3)Z_2Z_4 + \dots Z_2Z_m +
\dots \\ 
& \ \ \ \ \ \ \ \ \ \left. \left. + 2Z_{m-1}^2  +\cdots + Z_{m-1}Z_m + Z_m^2\} -
\{mZ_1^2 + (m-1)Z_2^2 + \dots + Z_m^2 \} \right) \bigg| M = m \right]
 \end{align*}
Using the expressions from (\ref{bi_ci}) and (\ref{xi_xj}) and noting that $Z_i
= X_i - \frac{1}{m}$, the above expression can be rewritten as,

\begin{align*}
T_m = \frac{2}{m} \bigg\{ \sum_{r = 0}^{m-1} [(m - r) (\rho^r c(1) +
b(1)b(r+1))] + \sum_{r = 0}^{m-1} [(m - 1 - r) (\rho^r c(2) + b(2)b(r+2))] +
\dots \bigg\} \\ - \frac{1}{m} \bigg\{ m[c(1) + b^2(1)] + (m-1)[c(2) + b^2(2)] +
\dots \bigg\}
\end{align*}
This can be condensed to be written as,
\begin{align}
\label{tm_final}
T_m = \frac{2}{m} \sum_{i=1}^m \sum_{r = 0}^{m-i} [(m - r -i + 1) (\rho^r c(i) +
b(i)b(r+i))] - \frac{1}{m} \sum_{i=1}^m (m-i+1)[c(i) + b^2(i)] 
\end{align}
Inspired by the above equation \eqref{tm_final}, define the following,  
\begin{align}\label{kl_def}
k_m := 2\sum_{i=1}^m \sum_{r = 0}^{m-i} \rho^r c(i) + b(i)b(r+i) \nonumber \\
l_m := \sum_{i=1}^m c(i) + b^2(i) =  \sum_{i=1}^m \mathbb{E}\left[ \left( X_i -
\frac{1}{m} \right)^2 \bigg| M = m \right] 
\end{align}
The last equality follows from (\ref{bi_ci}). The following inequalties are
noted. 
\begin{flalign}\label{kl_ineq}
& 2\sum_{i=1}^m \sum_{r = 0}^{m-i} [(m - r -i + 1) (\rho^r c(i) + b(i)b(r+i))]
\leq 2\sum_{i=1}^m \sum_{r = 0}^{m-i} m (\rho^r c(i) + b(i)b(r+i)) = mk_m
\nonumber \\
&  \ \ \ \ \ \ \ \ \ \ \quad \quad \ \ \ \ \ \ \ \ \ \ \ \sum_{i=1}^m
(m-i+1)[c(i) + b^2(i)] \geq \sum_{i=1}^m c(i) + b^2(i) = l_m 
\end{flalign}

Combining the expressions obtained in equations (\ref{tm_final}), (\ref{kl_def})
and (\ref{kl_ineq}), one can conclude $T_m \leq \frac{1}{m}(mk_m - l_m)$.  Again
consider the remainder term, $R_M = 1 - S_M$. Squaring both sides and taking
expectations, we get,
\begin{align}\label{sm_rm}
\mathbb{E}[(S_M - 1)^2 | M = m] = \mathbb{E}[R_M^2 | M = m]
\end{align}
For the left hand side term in the above equation,
\begin{align}
\label{sm_upper}
\mathbb{E}[(S_M - 1)^2 | M = m] & = \mathbb{E}\left[  \bigg\{\sum_{k = 1}^m
\left( X_k - \frac{1}{m} \right) \bigg\}^2 \bigg| M = m\right] \nonumber\\
& = \mathbb{E}\left[\left(\sum_{k = 1}^m Z_k\right)^2 \bigg| M = m\right]
\nonumber \\
& = \mathbb{E}\left[ Z_1^2 + Z_2^2 + \dots +Z_m^2 + 2\{ Z_1Z_2 + Z_1Z_3 + \dots
+ Z_{m-1}Z_m\} \bigg| M = m\right]\nonumber \\
& = \mathbb{E}\left[ 2\{Z_1^2 + Z_1Z_2 + \dots + Z_{m-1}Z_m\} - \{Z_1^2 + Z_2^2
\dots + Z_m^2\}\bigg| M = m\right] \nonumber\\
& = 2\sum_{i=1}^m \sum_{r = 0}^{m-i} (\rho^r c(i) + b(i)b(r+i)) - \sum_{i=1}^m
(c(i) + b^2(i)) \nonumber \\
& = k_m - l_m
\end{align}
For the term on the right hand side in \eqref{sm_rm}, 
\begin{align}\label{rm_upper}
R_M \leq \frac{\lambda}{n(1- \rho)} \implies \mathbb{E}[R_M^2 | M = m] \leq
\frac{\lambda^2}{n^2(1- \rho)^2}
\end{align}
The equations \eqref{sm_rm}, \eqref{sm_upper} and \eqref{rm_upper} can be
combined to give, 
\begin{align}
k_m - l_m \leq \frac{\lambda^2}{n^2(1- \rho)^2} \implies k_m \leq
\frac{\lambda^2}{n^2(1- \rho)^2} + l_m
\end{align}
Using the fact that $T_m \leq \frac{1}{m}(mk_m - l_m)$ and above inequality, the
upper bound on $T_m$ now becomes,
\begin{align}\label{tm_lm}
T_m  \leq \frac{1}{m}(mk_m - l_m) = k_m - \frac{l_m}{m} \leq
\frac{\lambda^2}{n^2(1- \rho)^2} + l_m - \frac{l_m}{m} = \frac{\lambda^2}{n^2(1-
\rho)^2} + l_m\left(1 - \frac{1}{m}\right)
\end{align}

Finally, taking off the conditional expectation and using expressions from
\eqref{tm_def}, \eqref{kl_def} and \eqref{tm_lm}, one can write, 
\begin{align}\label{second_last}
\mathbb{E}\left[ \frac{1}{M} \sum_{i = 1}^M \bigg| S_i - \frac{i}{M}
\bigg|^2\right] & \leq \frac{\lambda^2}{n^2(1- \rho)^2} + \mathbb{E}\left[
\left(\frac{M-1}{M}\right) \bigg\{ \sum_{i=1}^M  \left( X_i - \frac{1}{M}
\right)^2 \bigg\} \right] \nonumber\\ 
& \overset{\text{(a)}}{\leq} \frac{\lambda^2}{n^2(1- \rho)^2} + \mathbb{E}\left[
\bigg\{ \sum_{i=1}^M  \left( X_i - \frac{1}{M} \right)^2 \bigg\} \right]
\nonumber \\
%
%
& \overset{\text{(b)}}{\leq} \frac{\lambda^2}{n^2(1- \rho)^2} + \mathbb{E}\left[
2 \sum_{i=1}^M  \left( X_i^2 + \frac{1}{M^2} \right) \right] \nonumber\\ 
%
%
& \overset{\text{(c)}}{=} \frac{\lambda^2}{n^2(1- \rho)^2} + \mathbb{E}\left[  2
\sum_{i=1}^M  X_i^2 \right] + \mathbb{E}\left[\frac{2}{M} \right] \nonumber\\
& \overset{\text{(d)}}{\leq} \frac{\lambda^2}{n^2(1- \rho)^2} + 2
\mathbb{E}\left[ \sum_{i=1}^M  X_i^2 \right] + \frac{2\lambda}{n(1 - \rho) -
\lambda}
\end{align}
where (a) follows from the fact that $\dfrac{M -1}{M} < 1$, (b) is a direct
application of Cauchy Schwarz, (c) follows from linearity of expectation and (d)
uses the following result obtained from \eqref{M_upper} which can be restated as
$\dfrac{1}{M} < \dfrac{\lambda}{n(1 - \rho) - \lambda}$. Hence, 
\begin{equation*}
\mathbb{E}\left[\frac{1}{M} \right] < \frac{\lambda}{n(1 - \rho) - \lambda}
\end{equation*}
The proof will be complete if the term $\displaystyle \mathbb{E}\left[
\sum_{i=1}^M  X_i^2 \right]$ is upper bounded by a term of order $1/n$. For that
purpose, consider
\begin{equation*}
X_i = \sum_{r = 1}^i \rho^{i-r} Y_r \leq \sum_{r = 1}^i \rho^{i-r}
\frac{\lambda}{n} = \frac{\lambda}{n} \left(\frac{1 - \rho^i}{1- \rho}\right) 
\end{equation*}
This implies, $\displaystyle  X_i^2 \leq \frac{\lambda^2}{n^2} \left(\frac{1 -
\rho^i}{1- \rho}\right)^2 $. Taking expectations,
\begin{align}
\mathbb{E}\left[ \sum_{i=1}^M  X_i^2 \right] \leq  \mathbb{E}\left[ \sum_{i=1}^M
\frac{\lambda^2}{n^2} \left(\frac{1 - \rho^i}{1- \rho}\right)^2 \right] =
\frac{\lambda^2}{(n(1 - \rho))^2}  \mathbb{E}\left[ \sum_{i=1}^M (1 - \rho^i)^2
\right]
\end{align}
Since, $0 < \rho < 1$, therefore, $(1 - \rho^i)^2 < (1 - \rho^i) \ \forall \ i
\geq 1 $ and $\rho^i \geq \rho^j \ \forall \ j \geq i$. Hence, 
\begin{align}\label{exp_xi2}
\mathbb{E}\left[ \sum_{i=1}^M  X_i^2 \right] & \leq \frac{\lambda^2}{(n(1 -
\rho))^2}  \mathbb{E}\left[ \sum_{i=1}^M (1 - \rho^i)^2 \right] \nonumber\\
& \leq \frac{\lambda^2}{(n(1 - \rho))^2} \mathbb{E}\left[ \sum_{i=1}^M (1 -
\rho^i) \right] = \frac{\lambda^2}{(n(1 - \rho))^2} \left( \mathbb{E}[M] -
\mathbb{E}\left[ \sum_{i=1}^M \rho^i \right] \right) \nonumber \\
& \leq \frac{\lambda^2}{(n(1 - \rho))^2} \left( \mathbb{E}[M] - \mathbb{E}\left[
\sum_{i=1}^M \rho^M \right] \right) = \frac{\lambda^2}{(n(1 - \rho))^2} \left(
\mathbb{E}[M] -  \mathbb{E}\left[ M \rho^M \right] \right) 
\end{align}
To establish bounds on $\mathbb{E}\left[ M \rho^M \right]$, consider the
function $f: \mathbb{R} \rightarrow \mathbb{R}$, given by $f(x) = x\rho^x$,
where $0< \rho < 1$ is a finite constant independent of $x$. The function is
differentiable and its second derivative is given by $f^{''}(x) = \ln \rho ( 2 +
x\ln \rho ) \rho^x$. This is positive as long as $x > -\frac{2}{\ln \rho}$. Thus
the in region defined by $x > -\frac{2}{\ln \rho}$, the function is convex.
Since the function is convex in that interval, we can apply the Jensen's
inequality. Jensen's inequality states for any function $h(x)$ that is convex
over an interval $\mathcal{I} \subset \mathcal{D}$ where $\mathcal{D}$ is the
domain of $h(x)$, and a random variable $X$ whose support is a subset of
$\mathcal{I}$, the inequality $\mathbb{E}[h(X)] \geq h(\mathbb{E}[X])$ holds
true.

We know that $f$ is convex for $x > -\frac{2}{\ln \rho}$ and from
\eqref{M_upper}, we know that $M > \frac{n (1- \rho)}{\lambda} - 1$. Therefore
if $M > -\frac{2}{\ln \rho}$, then using Jensen's inequality, we can say that $
\mathbb{E}\left[ M \rho^M \right] \geq \mathbb{E}[M] \rho^{\mathbb{E}[M]}$. The
condition on $M$ will be always true if $\displaystyle \frac{n (1-
\rho)}{\lambda} - 1 > -\frac{2}{\ln \rho}$. That is if
\begin{align}\label{n_bound}
n > \frac{\lambda}{1- \rho} \left(1 - \frac{2}{\ln \rho} \right).
\end{align}
Therefore, if $n$ is large enough, we can write from \eqref{exp_xi2},
\begin{align}\label{last_bound}
\mathbb{E}\left[ \sum_{i=1}^M  X_i^2 \right] & \leq \frac{\lambda^2}{(n(1 -
\rho))^2} \left( \mathbb{E}[M] -  \mathbb{E}\left[ M \rho^M \right] \right)
\nonumber \\
& \leq \frac{\lambda^2}{(n(1 - \rho))^2} \left(\mathbb{E}[M] -  \mathbb{E}[M]
\rho^{\mathbb{E}[M]} \right) = \frac{\lambda^2}{(n(1 - \rho))^2}
\left(\mathbb{E}[M] ( 1 - \rho^{\mathbb{E}[M]} )  \right)\nonumber \\
& \overset{\text{(a)}}{\leq}  \frac{\lambda^2}{(n(1 - \rho))^2} \left( n +
\frac{(1- \rho)}{\lambda} \right) \left( 1 - \rho^{n + \frac{ 1}{\lambda}}
\right) 
\end{align}
where (a) follows from \eqref{EM_bounds} and decreasing nature of $\rho^x$. From
\eqref{EM_bounds}, we have $\mathbb{E}[M] \leq n + \frac{(1- \rho)}{\lambda} - 1
\leq n + \frac{(1- \rho)}{\lambda}$ and thus $1 - \rho^{\mathbb{E}[M]}  \leq 1 -
\rho^{n + \frac{(1- \rho)}{\lambda}} \leq 1 - \rho^{n + \frac{1}{\lambda}}$.

Using \eqref{second_last}, \eqref{exp_xi2}, \eqref{n_bound} and
\eqref{last_bound}, we can write
\begin{align}\label{final_bound_c0_c1}
\mathbb{E}\left[ \frac{1}{M} \sum_{i = 1}^M \bigg| S_i - \frac{i}{M}
\bigg|^2\right]  & \leq \frac{\lambda^2}{n^2(1- \rho)^2} + 2
\frac{\lambda^2}{(n(1 - \rho))^2} \left( n + \frac{(1- \rho)}{\lambda} \right)
\left( 1 - \rho^{n + \frac{ 1}{\lambda}} \right)  + \frac{2\lambda}{n(1 - \rho)
- \lambda} \nonumber \\
& \leq \frac{2\lambda^2 ( 1 - \rho^{n + \frac{1}{\lambda}})}{(1-
\rho)^2}\frac{1}{n} + \frac{2\lambda}{n(1 - \rho) - \lambda} +
\frac{\lambda^2}{n^2(1- \rho)^2}\left(1 + \frac{2}{\lambda} \right) \nonumber\\
& \leq \frac{C_0 (1 - C_1 \rho^n)}{n}
\end{align}
for some positive constants $C_0, C_1$ independent of $n$. This proves the upper
bound. It is very essential to note that this bound holds surely under the
condition that $\displaystyle n > \frac{\lambda}{1- \rho} \left(1 - \frac{2}{\ln
\rho} \right)$. Also, that is condition is a sufficient one, but not necessary.
This completes the proof.

\nocite{*}
\bibliographystyle{IEEE}

%

\end{document}